\documentclass{llncs}
\usepackage{amsmath, amssymb, stmaryrd, mathtools, mathrsfs, ntheorem}
\usepackage{graphicx, subfigure}
\usepackage{mathtools}
\usepackage{lineno}
\usepackage{soul}
\usepackage{hyperref, cleveref, xstring}
\crefrangeformat{figure}{~#3\StrGobbleRight{#1}{1}-#4#5\crefstripprefix{#1}{#2}#6}

\usepackage[english]{babel}
\usepackage{fancyhdr}
\hypersetup{
     colorlinks   = true,
     citecolor    = blue
}

\spnewtheorem*{example-rs}{Example}{\it}{\rm}

\usepackage[boxruled, vlined, linesnumbered]{algorithm2e}
\SetAlFnt{\small}
\SetAlCapFnt{\small}
\SetAlCapNameFnt{\small}

\SetCommentSty{mycommfont}

\usepackage{color}
\definecolor{fullred}{rgb}{0.85,.0,.1} 
\definecolor{orange}{rgb}{1,0.5,0}

\newcommand{\goth}[1]{\mathfrak{#1}}
\newcommand{\manifold}{\mathcal{M}}
\newcommand{\ball}[2]{\mathcal{B}(#1, #2)}
\newcommand{\horzdist}{\mathcal{H}}
\renewcommand{\natural}{\mathbb{N}}
\newcommand{\real}{\mathbb{R}}
\newcommand{\kdtree}{\mbox{$k$-d} tree\xspace}
\newcommand{\kdtrees}{\mbox{$k$-d} trees\xspace}
\newcommand{\range}{\text{\texttt{rng}$\,$}}

\title{Efficient Nearest-Neighbor Search for Dynamical Systems with Nonholonomic Constraints\vspace{-3mm}}
\author{Valerio Varricchio \and Brian Paden \and Dmitry Yershov \and Emilio Frazzoli\vspace{-2mm}}
\institute{Massachusetts Institute of Technology}
\date{\today }

\begin{document}
\graphicspath{{./fig/}}
\maketitle
\thispagestyle{plain} 
\vspace{-6mm}
 \begin{abstract}
Nearest-neighbor search dominates the asymptotic complexity of sampling-based motion planning algorithms and is often addressed with \kdtree data structures.
While it is generally believed that the expected complexity of nearest-neighbor queries is $O(\log(N))$  in the size of the tree, this paper reveals that when a classic \kdtree approach is used with sub-Riemannian metrics, the expected query complexity is in fact $\Theta(N^p \log(N))$ for a number $p \in [0, 1)$ determined by the degree of nonholonomy of the system.
These metrics arise naturally in nonholonomic mechanical systems, including classic wheeled robot models.
To address this negative result, we propose novel \kdtree build and query strategies tailored to sub-Riemannian metrics and demonstrate significant improvements in the running time of nearest-neighbor search queries.
\end{abstract}

\section{Introduction}
{
Sampling-Based algorithms such as Probabilistic Roadmaps (PRM) \cite{kavraki1996probabilistic}, Rapidly exploring Random Trees (RRT) \cite{lavalle1998rapidly} and their asymptotically optimal variants (PRM$^*$, RRT$^*$) \cite{karaman2011sampling} are widely used in motion planning. These algorithms build a random graph of motions between points on the robot's configuration manifold.

During the graph expansion, nearest-neighbor search is used to limit the computation to regions of the graph close to the new configurations and it is shown to dominate the asymptotic complexity of randomized planners.
The notion of closeness appropriate for motion planning is induced by the length of the shortest paths between configurations, or in general, by the minimum cost of controlling a system between states.

As compared to exhaustive linear search, efficient algorithms with reduced complexity have been studied in computational geometry and their use in motion planning has been highlighted as a factor of dramatic performance improvement \cite{lavalle2001randomized}. Among a variety of approaches, \kdtrees \cite{bentley1975multidimensional} are ideal due to their remarkable efficiency in low-dimensional spaces, typical of motion planning problems.

Classic \kdtrees are shown to have logarithmic average case complexity for distance functions \emph{strongly equivalent} to L-p metrics \cite{friedman1977algorithm}.
While this requirement is reasonable for many applications, it does not apply to distances induced by the shortest paths of nonholonomic systems.
Therefore, identifying nearest neighbors in the sense of a generic control cost remains an important open problem in sampling-based motion planning.
In both literature and practical implementations, when searching for neighbors, randomized planners resort to distance functions that only approximate the true control cost.
Arguably, the most common choices are Euclidean distance or quadratic forms~\cite{lavalle2001randomized}.
This ad-hoc approach can significantly slow down the convergence rate of sampling-based algorithms if an inappropriate metric is selected.

A number of heuristics have been proposed to resolve this issue, such as the \emph{reachability} and \emph{utility guided RRTs}~\cite{shkolnik2009reachability,burns2007single} which bias the tree expansion towards promising regions of the configuration space.
Other approaches use the cost of linear quadratic regulators \cite{glassman2010quadratic} and learning techniques \cite{palmieri2014distance,bharatheesha2014distance}.
In specific examples, these heuristics can significantly reduce the negative effects of finding nearest neighbors according to a metric inconsistent with the minimum cost path between configurations.
However, the underlying inconsistency is not directly addressed. In contrast, a strong motivation to address it comes from recent research \cite{karaman2013sampling}, which shows that a major speedup of sampling-based kynodinamic planners can be achieved by considering nonholonomy at the stage of  nearest-neighbor search.

Specialized \kdtree algorithms have been proposed to account for non-standard topologies of some configuration manifolds \cite{kuffner2004effective,yershova2007improving,ichnowski2015fast}.

However, no effort is known towards generalizing such algorithms to differential constraints.

In this work, we investigate the use of \kdtrees for \emph{exact} nearest-neighbor search in the presence of differential constraints.
The main contributions can be summarized as follows: (i) we derive the expected complexity of nearest-neighbor queries with \kdtrees built according to classic techniques and reveal that it is super-logarithmic (ii) we propose novel \kdtree build and query procedures tailored to sub-Riemannian metrics (iii) we provide numerical trials which verify our theoretical analysis and demonstrate the improvement afforded by the proposed algorithms as compared with popular open source software libraries, such as FLANN \cite{muja_flann_2009} and OMPL \cite{sucan2012the-open-motion-planning-library}.

In Section \ref{sec:geometry_review}, we review background material on sub-Riemannian geometries and show connections with nonholonomic systems, providing asymptotic bounds to their reachable sets.
Based on these bounds, in Section \ref{sec:kdtrees} we propose a query procedure specialized for nonholonomic systems, after a brief review of the \kdtree algorithm.
In Section \ref{sec:analysis}, we study the expected complexity of $m$-nearest-neighbor queries on a classic \kdtree with sub-Riemannian metrics.
Inspired by this analysis, in Section \ref{sec:lie_split} we propose a novel incremental build procedure. Finally, in Section \ref{sec:experiments} we show positive experimental results for a nonholonomic mobile robot, which confirm our theoretical predictions and the effectiveness of the proposed algorithms.

\section{Geometry of Nonholonomic Systems} \label{sec:geometry_review}
Nonholonomic constraints are frequently encountered in robotics and describe mechanical systems whose local mobility is, in some sense, limited.
Basic concepts from differential geometry, reviewed below, are used to clarify these limitations and discuss them quantitatively.

\subsection{Elements of Differential Geometry}
A subset $\mathcal{M}$ of $\mathbb{R}^n$ is a \emph{smooth $k$-dimensional manifold} if for all $p\in \manifold$ there exists a neighborhood $V$ of $\manifold$ such that $V\cap \manifold$ is diffeomorphic to an open subset of $\mathbb{R}^k$.
A vector $v \in \real^n$ is said to be \emph{tangent to $\manifold$ at point $p \in \manifold$} if there exists a smooth curve $\gamma: [0,1] \rightarrow \manifold$, such that $\dot \gamma (0) = v$ and $\gamma(0) = p$.
The \emph{tangent space} of $\mathcal{M}$ at $p$, denoted $T_p\manifold$, is the subspace of vectors tangent to $\manifold$ at $p$. %
A map $Y: \manifold \rightarrow \mathbb{R}^n$ is a \emph{vector field} on $\manifold$ if for each $p\in \mathcal{M}$, $Y(p) \in T_p\manifold$.
A smooth euclidean \emph{vector bundle} of rank $r$ over $\manifold$ is defined as a set $E \subset \manifold \times \real^l$ such that the set $E_p:\left\{v\in \real^l, (p,v) \in E\right\}$ is an $r$-dimensional vector space for all $p \in \manifold$.
$E_p$ is called the \emph{fiber} of bundle $E$ at point $p$.
Any set of $h$ linearly independent smooth vector fields $Y_1, \dots Y_h$ such that for all $p \in \manifold,$ $span(Y_1(p), ..., Y_h(p)) = E_p$ is called a \emph{basis} of $E$ and $Y_i$ are called \emph{generator vector fields} of $E$.
The \emph{tangent bundle} of $\manifold$ is defined as the vector bundle $T\manifold$ whose fiber at each point $p$ is the tangent space at that point, $T_p\manifold$. A \emph{distribution} $\horzdist$ on a manifold is a subbundle of the tangent bundle, i.e., a vector bundle such that its fiber $\horzdist_p$ at all points is a vector subspace of $T_p\manifold$.
\vspace{-2mm}
\paragraph{Connection with nonholonomic systems.}
Consider a nonholonomic system described by the differential constraint:
\begin{equation}\label{eq:dynamics}
\dot x (t)= f(x(t), u(t)),
\end{equation}
where the configuration $x(t)$ and control $u(t)$ belong to the \emph{smooth manifolds} $\mathcal X$ and $\mathcal U$ respectively.
Each ${\goth u_i}\in \mathcal{U}$ defines a \emph{vector field} $g_i(z) = f(z, \goth u_i)$ on $\mathcal X$.
Therefore, for a fixed configuration $z$, $f(z, u)$ has values in a vector space $\horzdist_{z} \coloneqq Span(\{g_i(z)\}) $.
In other words, the dynamics described by equation (\ref{eq:dynamics}) define a \emph{distribution} $\horzdist$ on the configuration manifold.

\vspace{2mm}
Throughout the paper, we will make use of the Reeds-Shepp vehicle~\cite{reeds1990optimal} as an illustrative example.
\begin{example-rs}[Configuration manifold of a Reeds-Shepp vehicle]
The configuration manifold for this vehicle model is $\mathcal X = SE(2)$ with coordinates $x=(x_1,x_2,x_3)$.
The mobility of the system is given by $\dot x_1 = u_1 \cos(x_3),\, \dot x_2 = u_1 \sin(x_3),\, \dot x_3 = u_1 u_2,$ where $u_1,u_2\in [-1,1]$.
The inputs $\goth u_1 = (1,0)$ and $\goth u_2 = (1,1)$ define the vector fields $g_1(x) = (\cos(x_3), \sin(x_3), 0)$ and $ g_2(x) = (\cos(x_3), \sin(x_3), 1)$.
At each $z\in SE(2)$ the fiber of the Reeds-Shepp distribution $\horzdist^{RS}$ is $Span\{g_1(z),$ $ g_2(z)\}$.
Let $\hat f(x) = (\cos(x_3), \sin(x_3), 0)$, $\hat l(x) = (-\sin(x_3), \cos(x_3), 0)$ and $\hat \theta(x) = (0,0,1)$. These vector fields indicate the body frame of the vehicle, i.e., its \emph{front} $\hat f$, \emph{lateral} $\hat l$ and \emph{rotation} $\hat \theta$ axes. The Reeds-Shepp distribution is then equivalently defined by the fibers
 $\horzdist^{RS}_{z} = Span\{\hat f(z), \hat \theta(z)\}$.
\end{example-rs}

\subsection{Distances in a Sub-Riemannian Geometry}
A \emph{sub-Riemannian geometry} $\mathcal G$ on a manifold $\manifold$ is a tuple $\mathcal G =(\manifold, \horzdist, \langle\cdot,\cdot\rangle_\horzdist)$, where $\horzdist$ is a distribution on $\manifold$ whose fibers at all points $p$ are equipped with the inner product $\langle\cdot, \cdot\rangle_\horzdist : \horzdist_p \times \horzdist_p \rightarrow \real$.
The distribution $\horzdist$ is referred to as the \emph{horizontal distribution}.
A smooth curve $\gamma: [0,1] \rightarrow \manifold$ is said to be \emph{horizontal} if $\dot \gamma(t) \in \horzdist_{\gamma(t)}$ for all $t \in [0,1]$.
The length of smooth curves is defined $\ell_{\mathcal G}(\gamma):=\int_0^1 \sqrt{\langle\dot \gamma(t),\dot \gamma(t)\rangle_\horzdist}\,  dt$.
If $\Gamma_a^b$ denotes the set of horizontal curves between $a,b\in\manifold$, then $d_\mathcal{G}(a,b)\coloneqq\inf_{\gamma\in\Gamma_a^b}\ell_{\mathcal G}(\gamma)$ is a \emph{sub-Riemannian metric} defined by the geometry.
The ball centered at $p$ of radius $r$ with respect to the metric $d_\mathcal{G}$ is denoted  $\ball{p}{r} $.
In control theory, the \emph{attainable set} $\mathcal{A}(x_0,t)$ is the subset of $\mathcal{X}$ such that for each $p\in\mathcal{A}(x_0,t)$ there exists a $u:[0,\tau]\rightarrow \mathcal{U}$ for which the solution to (\ref{eq:dynamics}) through $x(0)=x_0$ satisfies $x(\tau)=p$ for some $\tau\leq t$.
Solutions to (\ref{eq:dynamics}) are simply horizontal curves, so the attainable set $\mathcal A(x_0, t)$ is equivalent to the ball $\mathcal B(x_0, t)$ defined by the sub-Riemannian metric. This equivalence holds for systems with time-reversal symmetry, under the controllability condition stated in Theorem~\ref{thm:chow}.

\begin{example-rs}[Sub-Riemannian geometry of a Reeds-Shepp vehicle]
The geometry associated with the Reeds-Shepp vehicle is $\mathcal G^{RS} = (SE(2), \horzdist^{RS}, \langle\cdot, \cdot\rangle_{RS}) $, with the standard inner product $\langle v, w\rangle_{RS} = v^Tw$. Horizontal curves for this geometry are feasible paths satisfying the differential constraints. Geodesics correspond to minimum-time paths between two configurations and are known in closed form~\cite{reeds1990optimal}.
\end{example-rs}
\subsection{Iterated Lie Brackets and the Ball-box Theorem}
A system is said to be \emph{controllable} if any pair of configurations can be connected by a feasible (horizontal) path.
Determining controllability of a nonholonomic system is nontrivial.
For example, the Reeds-Shepp vehicle cannot move directly in the lateral direction, but intuition suggests that an appropriate sequence of motions can result in a lateral displacement (e.g., in parallel parking).
Chow's Theorem and the Ball-box Theorem, reviewed below, are fundamental tools related to the controllability and the reachable sets of nonholonomic systems.

The \emph{Lie derivative} of a vector field $Y$ at $p\in\manifold$ in the direction $v \in T_p\manifold$ is defined as $dY(p)v  = \left. \frac{d}{dt} Y(\gamma(t))\right|_{t=0}$, where $\gamma$ is a smooth curve starting in $p = \gamma(0)$ with velocity $v = \dot \gamma(0)$. Given two vector fields  $Y_1,Y_2$ on $\manifold$, the \emph{Lie bracket} $[Y_1, Y_2]$  is a vector field on $\manifold$ defined as $[Y_1, Y_2](p) = dY_2(p)Y_1(p) -  dY_1(p)Y_2(p)$.

From the horizontal distribution $\horzdist$, one can construct a sequence of distributions by iterating the Lie brackets of the generating vector fields, $Y_1, Y_2 \dots Y_h$, with $h \leq k$. Recursively, this sequence is defined as: $\horzdist^1 = \horzdist$, $\horzdist^{i+1} = \horzdist^i \cup [\horzdist, \horzdist^i]$,
where $[\horzdist, \horzdist^i]$ denotes the distribution given by the Lie brackets of each generating vector field of $\horzdist$ with those of $\horzdist^i$.
Note that the Lie bracket of two vector fields can be linearly independent from the original fields, hence $\horzdist^i\subseteq \horzdist^{i+1}$. The \emph{Lie hull}~\cite{montgomery2006tour}, denoted $Lie(H)$, is the limit of the sequence $\horzdist^i$ as $i \rightarrow \infty$. A distribution $\horzdist$ is said to be \emph{bracket-generating}  if $Lie(\horzdist) = T\manifold$.

\begin{theorem}{\bf (Chow's Theorem }\cite[p. 44]{montgomery2006tour} {\bf )}.\label{thm:chow} If $Lie(\horzdist) = T\manifold$ on a connected manifold $\manifold$, then any $a, b \in \manifold$ can be joined by a horizontal curve. \end{theorem}

\begin{example-rs}[Lie hull of a Reeds-Shepp vehicle]
\label{ex:chow}
Consider the Lie hull of $\horzdist^{RS}$ generated by $\{ \hat{f}(x),\hat{ \theta}(x) \}$.
For every $x\in SE(2)$, these vectors span a two-dimensional subspace of $T_x\mathcal  (SE(2))$.
The first order Lie bracket is given by
\begin{equation}
 [\hat f, \hat \theta](x) = \left(\frac{d}{dx}\hat f(x)\right) \hat{\theta}(x) - \left(\frac{d}{dx}\hat{\theta}(x)\right) \hat{f}(x)= \left(-\sin(x_3), \cos(x_3), 0 \right).
\end{equation}
This coincides with the body frame lateral axis $\hat l(x)$ of the vehicle.
Therefore, the second order distribution $\horzdist^2_x = Span\{\hat f(x), \hat \theta(x), [\hat f, \hat \theta](x)\}$ spans the tangent bundle of $SE(2)$, and thus the Lie hull is obtained in the second step.
By Theorem \ref{thm:chow}, there is a feasible motion connecting any two configurations which is consistent with one's intuition about the wheeled robot.
\end{example-rs}

Above, we have shown that from a basis $Y_1 \dots Y_h$ of a bracket-generating distribution $\horzdist$, one can define a basis $y_1, \dots y_k$ of $T\manifold$. Specifically $y_i = Y_i$ for $i\leq h$, while the remaining $d-h$ fields are obtained with Lie brackets.
The vector fields $\{y_i\}$ are called \emph{privileged directions} \cite{montgomery2006tour}.
Define the \emph{weight} $w_i$ of the privileged direction $y_i$ as the smallest order of Lie brackets required to generate $y_i$ from the original basis. More formally, $w_i$ is such that $y_i \notin \horzdist^{w_i-1}$ and $y_i \in \horzdist^{w_i}$ for all $i$. The \emph{weighted box} at $p$ of size $\epsilon$, weights $w \in \natural^k_{>0}$ and multipliers $\mu \in \real^k_{> 0}$ is defined as:
\begin{equation}\mathrm{Box}^{w,\mu}(p, \epsilon): \{y \in \real^n: |\langle y-p, y_i(p)\rangle_{\real^n}|< \mu_i \epsilon^{w_i} \text{ for all } i \in [1, k] \}. \end{equation}
\begin{theorem}[The ball-box Theorem]\label{thm:ball_box}
Let $\horzdist$ be a distribution on a manifold $\manifold$ satisfying the assumptions of Chow's Theorem. Then, there exist constants $\epsilon_0 \in \real_{>0}$ and $c, C \in \real^k_{>0}$ such that for all $\epsilon < \epsilon_0$ and for all $p \in \manifold$:
\begin{equation}\label{eq:ball_box}
\mathrm{Box}^w_{i}(p, \epsilon) \coloneqq \mathrm{Box}^{w, c}(p, \epsilon) \subset \ball{p}{\epsilon} \subset \mathrm{Box}^{w, C}(p, \epsilon) \coloneqq \mathrm{Box}^w_{o}(p, \epsilon),
\end{equation}
where $\mathrm{Box}^w_i$ and $\mathrm{Box}^w_o$ are referred to as the inner and outer bounding boxes for the sub-Riemannian ball $\mathcal B$, respectively.
\end{theorem}

\begin{example-rs}[Ball-Box Theorem visualized for a Reeds-Shepp vehicle]
From the Lie hull construction in the previous example, we know that $\hat f(x), \hat \theta(x) \in \horzdist^1$ so the corresponding weights are $w_{\hat f} = w_{\hat \theta} = 1$. Conversely, $\hat l(x)$ first appears in $\horzdist^2$, so its weight is $w_{\hat l} = 2$.
Theorem \ref{thm:ball_box} states the existence of inner and outer bounding boxes for reachable sets as $t \rightarrow 0$ and predicts the infinitesimal order of each side of these boxes, as shown in figure \ref{fig:ball_box}.
Higher order Lie brackets correspond to sides that approach zero at a faster asymptotic rate. The longitudinal and angular sides --- along $\hat f$ and $\hat \theta$ --- scale with $\Theta(t)$, while the lateral one --- along $\hat l$ --- with $\Theta(t^2)$. Therefore, both boxes become increasingly elongated along $\hat f$ and $\hat \theta$ and flattened along $\hat l$ as $t\shortrightarrow 0$.
Intuitively, this geometric feature of boxes reflects the well known fact that a small lateral displacement of a car requires more time than an equivalent longitudinal one.
\end{example-rs}
  \begin{figure}[!ht]\vspace{-3mm}\hspace{-0.2cm}
    \includegraphics[width=\textwidth, page=2]{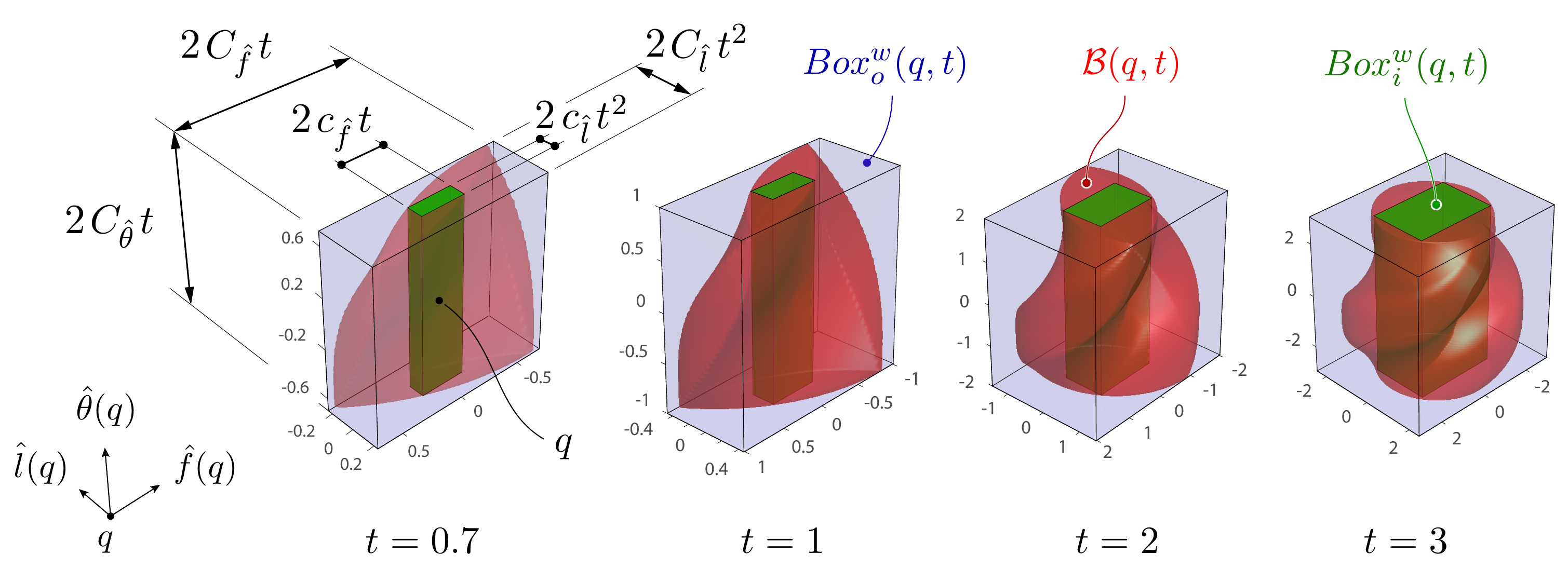}
    \caption{Reachable sets and bounding boxes for a Reeds-Shepp vehicle around a configuration $q$ and for different values of $t$. The lengths of the box sides are highlighted. For both the inner and outer boxes, as $t\rightarrow 0$, the sides along the front $\hat f$ and heading $\hat \theta$ axes are linear in $t$, while the side along the lateral axis $\hat l$ is quadratic in $t$.\vspace{-4mm}}
    \label{fig:ball_box}
  \end{figure}
For the Reeds-Shepp vehicle, the sides of both boxes can be computed explicitly with geometric considerations on special manoeuvers that maximize or minimize the displacement along each axis.
 In particular, we get the values $ C_{\hat f}~=~C_{\hat \theta}~=~c_{\hat \theta}~=~1,\; C_{\hat l}~=~1/2,  \; c_{\hat f} = \sqrt{3/2}-1, \; c_{\hat l} = 1/8$.

\section{The \kdtree Data Structure} \label{sec:kdtrees}
A \kdtree $\mathcal{T}$ is a binary tree organizing a \emph{finite} subset $X \subset \manifold$, called a \emph{database}, with its elements $x_i \in X$ called \emph{data points}.
We would like to find the $m$ points in $X$ closest to a given \emph{query point} $q$ on the manifold. 
Each point $x_i\in X$ is put in relation with a normal vector $n_i\in \real^n$.
Together, the pair $v_i = (x_i,n_i)$ defines a \emph{vertex} of the binary tree.
The set of vertices is denoted with $\mathcal{V}$.
A vertex defines a partition of $\mathbb{R}^{n}$ into two halfspaces, referred to as the \emph{positive} and \emph{negative halfspace}, and respectively described algebraically:
\begin{equation}
\begin{array}{l}
\mathfrak{h}^{+}(x_i, n_i)\coloneqq\left\{ x\in\mathbb{R}^{n}:\,\left\langle n_i,x\right\rangle _{\mathbb{R}^{n}}>\left\langle n_i,x_i\right\rangle _{\mathbb{R}^{n}}\right\} ,\\
\mathfrak{h}^{-}(x_i, n_i)\coloneqq\left\{ x\in\mathbb{R}^{n}:\,\left\langle n_i,x\right\rangle _{\mathbb{R}^{n}}\leq\left\langle n_i,x_i\right\rangle _{\mathbb{R}^{n}}\right\}.
\end{array}
\end{equation}
An \emph{edge} is an ordered pair of vertices $e = (v_i, v_j)$.
A binary tree is defined as  $\mathcal T\coloneqq(\mathcal V, \mathcal E^-, \mathcal E^+)$, with $\mathcal V$ set of vertices, $\mathcal E^-$ set of left edges and $\mathcal E^+$ set of right edges.
Given one edge $e = (v_i, v_j)$, vertex $v_j$ is referred to as the \emph{left child} of $v_i$ if $e \in \mathcal E^-$, or the \emph{right child} if $e \in \mathcal E^+$.
Let $\verb|parent|(v_i\in \mathcal V) = v_j \in \mathcal V$ s.t. $(v_i, v_j) \in \mathcal E^- \cup \mathcal E^+$. By convention, $\verb|parent|(v_i) =\emptyset$ if such $v_j$ does not exist and $v_i$ is called the \emph{root} of $\mathcal T$, denoted $v_i =\verb|root|(\mathcal T)$.
Let  $\verb|child|(v_i \in \mathcal V, s \in \{-, +\}) = v_j \in \mathcal V$ s.t. $(v_i, v_j) \in \mathcal E^s$, or otherwise $\emptyset$ if such $v_j$ does not exist.

The fundamental property of \kdtrees is that \emph{left children belong to the negative halfspace defined by their parents and right children to the positive halfspace}.
Recursively, a vertex belongs to the parent halfspaces of all its ancestors.
As a result, a \kdtree defines a \emph{partition} of $\manifold$ into non-overlapping polyhedra, called \emph{buckets} \cite{friedman1977algorithm}, that cover the entire manifold. In the sequel, we let $\mathscr B_{\mathcal T }$ denote the set of buckets for a given \kdtree $\mathcal T$.

Buckets are associated with the leaves of $\mathcal T$ and with parents of only-child leaves (i.e., leaves without a sibling). For any given point $q\in \manifold$, we denote with $\goth b_q$ the unique bucket containing $q$.

 \SetKwFunction{leftSubset}{leftSubset}
 \SetKwFunction{median}{median}

\subsection{$m$-nearest-neighbor query algorithm} \label{sec:query}

The computational efficiency afforded by the algorithm comes from the following observation:
let $q$ be the query point and suppose that among the distance evaluations computed thus far, the $m$ closest points to $q$ are within a distance $d_m$ away.
Now consider a vertex $(x,n)$ and suppose the query point $q$ is contained in $\mathfrak{h}^{+}(x,n)$ and $\mathcal B(q, d_m)\subset\mathfrak{h}^{+}(x,n)$.
The data points represented by the sibling of $(x,n)$ and all of its descendants are contained in $\mathfrak{h}^{-}(r,c)$ and are therefore at a distance greater than $d_m$.
Thus, the corresponding subtree can be omitted from the search.
The following primitives will be used to define the query procedure presented in Algorithm \ref{alg:query}.

\vspace{2mm} \noindent{\bf Side of hyperplane.} A procedure \verb|sideOf|$(x, p, n) \rightarrow \{-,+\}$, with $x, p \in \manifold$, $n \in \real^n$. Returns + iff $x \in \goth h^+(p,n)$ and $-$ otherwise. For convenience, define $\verb|opposite|(+) = -$ and vice-versa.

\vspace{2mm} \noindent{\bf Queue.} For an $m$-nearest-neighbor query, the algorithm maintains a \emph{Bounded Priority Queue}, $Q$ of size $m$ to collect the results. The queue can be thought of as a sequence $Q = [q_1, q_2, \dots q_m ]$, where each element $q_i$ is defined as a distance-vertex pair $q_i: (d_i, v_i), \, d_i \in \real_{\geq 0}, \, v_i \in \mathcal V$. The property $d_1 \leq d_2 \leq \dots \leq d_m$ is an invariant of the data structure.
When an element $(d_{new}, n_{new})$ is inserted in the queue, if $d_{new} < d_m$, then $q_m$ is discarded and the indices of the remaining elements are rearranged to maintain the order.

\vspace{2mm} \noindent{\bf Ball-Hyperplane Intersection.} A procedure that determines whether a ball and a hyperplane  intersect. Precisely, \verb|ballHyperplane|$(x, R, p, n)$, with $x \in \mathcal M$ and $R \geq 0$, returns true if $\ball{x}{R} \cap \goth h(p,n) \neq \emptyset$. Note that it does not need to return false otherwise.\vspace{1.5mm}

In the classic analysis of \kdtrees~\cite[eq. 14]{friedman1977algorithm} the distance is assumed to be a sum of component-wise terms, called \emph{coordinate distance functions}. Namely:
\vspace{-.5mm}
\begin{equation}\label{eq:accum_dist} \textstyle
d(x, y) = F\left(\sum_{i=1}^k d_i(x_i, y_i)\right).\vspace{-.5mm}
\end{equation}
When this holds, e.g., for L-p metrics, the ball-hyperplane intersection procedure reduces to the direct comparison of two numbers.

However, this property does not hold for other notions of distance and in general sub-Riemannian balls can have nontrivial shapes. For these cases, the procedure can be implemented by checking that $V_o(x, R) \,\cap\, \goth h(p,n) \neq \emptyset$, where $V_o$ is an ``outer set'' with a convenient geometry, i.e., a set such that $ V_o(x, R) \supseteq \mathcal B(x, R)$ for all $x \in \manifold, R>0$ and such that intersections with hyperplanes are easy to verify.
Clearly, $\mathrm{vol}(V_o(x, R))/\mathrm{vol}(\mathcal B(x, R)) \geq 1$ and it is desirable that this ratio stays as small as possible.

A crucial consequence of Theorem \ref{thm:ball_box} is that the choice $V_o(x, R) = \mathrm{Box}_o^w(x, R)$ offers \emph{optimal asymptotical behavior}, since the volume of the outer set scales with the smallest possible order, namely: $\mathrm{vol}(\mathrm{Box}^w_o(x, R)) \in \Theta[\mathrm{vol}(\mathcal B(x, R))]$. This fact is fundamental to devise efficient query algorithms for nonholonomic systems.

\begin{example-rs}[Ball-hyperplane intersection for a Reeds-Shepp vehicle]
  The reachable sets shown in figure \ref{fig:ball_box} have a nontrivial geometry. Let us consider two possible implementations of the ball-hyperplane intersection procedure:
  \begin{enumerate}
    \item \emph{Euclidean bound} (EB). Define the set $\mathcal C(p, R) = \{x \in SE(2)\,| (x_1-p_1)^2+(x_2-p_2)^2\leq R^2, |x_3-p_3|\leq 2R\}$, i.e., a cylinder in the configuration space with axis along $\hat \theta$.
    Since the Euclidean distance $\|a-b\|$ is a lower bound to $d_{RS}(a, b)$, then $\mathcal B(x, R) \subset \mathcal C(x, R)$ and \verb|ballHyperplane| can be correctly implemented by checking $\mathcal C(x,R) \cap h(p,n) \neq \emptyset$ \vspace{1mm}
    \item \emph{Outer Box Bound} (BB).
    In this case, the procedure checks $\mathrm{Box}^w_o(x, R) \cap h(p,n) \neq \emptyset$. A simple implementation is to test whether all vertices of the box (8 in this case) are on the same halfspace defined by $\goth h(p,n)$.
  \end{enumerate}
While Euclidean bounds (EB) are a simple choice, the Outer Box Bound (BB) method approximates the sub-Riemannian ball tightly.
In fact, as $R\rightarrow 0$, $\mathrm{vol}(\mathcal C(p, R)) \in \Theta(R^3)$, while $\mathrm{vol}(\mathrm{Box}_o^w(p, R)) \in \Theta(R^4)$ and therefore the volume of the cylinder tends to be infinitely larger than the volume of the outer box.
As a result, method (BB) yields an asymptotically unbounded speedup in the algorithm compared to (EB), as confirmed by our experimental results in Section \ref{sec:experiments}. In addition, any other implementation of the procedure will at most provide a constant factor improvement with respect to method (BB).
\vspace{-7mm}
\end{example-rs}
 \begin{algorithm}[!ht]
   \label{alg:query}
   \caption{\kdtree query}
 \SetKwData{Q}{$Q$}
 \SetKwData{side}{side}
 \SetKwFunction{dist}{dist}
 \SetKwFunction{assign}{\ensuremath{\longleftarrow}}
 \SetKwFunction{sideOf}{sideOf}
 \SetKwFunction{child}{child}
 \SetKwFunction{ballHyperplane}{ballHyperplane}
 \SetKwFunction{query}{query}
 \SetKwFunction{rootfcn}{root}
 \SetKwFunction{opposite}{opposite}
 \SetKwFunction{querySubtree}{querySubtree}
 \SetKwBlock{qfnblock}{define \query($q \in \manifold$, $\mathcal T = (\mathcal V, \mathcal E^-, \mathcal E^+)$):}{end}
 \SetKwBlock{qsfnblock}{define \querySubtree($q \in \manifold$, $v_i = (x_i, n_i) \in \mathcal V$):}{end}

 \qfnblock{
 $Q$ \assign $\{q_{1:m} = (\emptyset, \infty)\}$ \tcp*{initialize queue}
 \qsfnblock{
   \lIf(\tcp*[f]{caller reached leaf}){$v_i = \emptyset$}{
     \Return
   }
   $s$ \assign \sideOf($q$, $x_i$, $n_i$)\;
   \querySubtree($q$, $\child(v_i, s)$) \tcp*{descend to child containing $q$}

   Q \assign $Q \cup \{(\dist(q, \, x_i), v_i) \}$ \tcp*{add vertex to the queue}

   \If(\tcp*[f]{check for intersections}){\ballHyperplane($q$, $d_k$, $x_i$, $n_i$)}{
     \querySubtree($q$, $\child(v_i, \opposite(s)))$ \tcp*{visit sibling}
   }
 }

 \querySubtree($q$, \rootfcn$ (\mathcal T)$) \tcp*{start recursion from root}
 \Return $Q$\;
 }
 \end{algorithm}

\vspace{-10mm}\subsection{\kdtree build algorithms}\label{sec:build}

The performance of nearest-neighbor queries is heavily dependant on how the \kdtree is constructed.
In the sequel, we describe two popular approaches to construct \kdtrees: \emph{batch} (or static) and \emph{incremental} (or dynamic).

In the batch algorithm, all data points are processed at once and statistics of the database determine the vertices of the tree at each depth. This algorithm guarantees a balanced tree, but the insertion of new points is not possible without re-building the tree. Conversely, in the incremental version, the tree is updated on the fly as points are inserted, however, the tree balance guarantee is lost.

Both algorithms find applications in motion planning: the batch version can be used to efficiently build roadmaps for off-line, multiple-query techniques such as PRMs, while the incremental is suitable for anytime algorithms such as RRTs.\vspace{-2mm}

\subsection*{Batch \kdtree build algorithm}\vspace{-2mm}
The following primitive procedures will be used to describe the batch construction of \kdtrees, defined in Algorithm \ref{alg:batch}.
The range of a set $S$ along direction $\hat d$ is defined as $\range_{\hat d}(S) = \sup_{x \in S}\langle x, \hat d \rangle  - \inf_{x \in S}\langle x, \hat d \rangle $.

\vspace{2mm} \noindent{\bf Maximum range.} Let $\verb|maxRange|: 2^X \rightarrow \real^n$. Given a subset $D \in 2^X$ of the data points, \verb|maxRange| determines a direction $l$ along which the range of data is maximal. We consider the formulation in \cite{friedman1977algorithm}, where $l$ is chosen among the cardinal directions $\{\hat e_i\}_{i\in[0, \dots k-1]}$, so that $\verb|maxRange|(D) = \arg\max_{\hat e_i}[ \range_{\hat e_i}( D )]$.

\vspace{2mm} \noindent{\bf Median element.} Let $\verb|median|: 2^X \times \real^n \rightarrow D$. Given a subset $D \in 2^X$ of the data points and a direction $l$, \verb|median| returns the median element of $D$ along direction $l$.
\vspace{-2mm}
 \begin{algorithm}[!ht]\label{alg:batch}
   \caption{Batch \kdtree build}
 \SetKwFunction{assign}{\ensuremath{\longleftarrow}}
 \SetKwFunction{build}{build}
 \SetKwFunction{buildSubtree}{buildSubtree}
 \SetKwFunction{maxRange}{maxRange}
 \SetKwFunction{leftSubset}{leftSubset}
 \SetKwFunction{median}{median}
 \SetKwBlock{bfnblock}{define \build($X$):}{end}

 \SetKwBlock{bsfnblock}{define \buildSubtree($D \in 2^X$, $v_i = (x_i, n_i) \in \mathcal V$, $s \in \{-,+\}$):}{end}
 \bfnblock{
   $\mathcal V\assign \emptyset; \,\mathcal E^-\assign \emptyset; \, \mathcal E^+ \assign \emptyset$ \tcp*{Initialize tree}
   \bsfnblock{
   \lIf(\tcp*[f]{caller reached leaf}){$D = \emptyset$}{
   \Return
   }
   $n_{new}$ \assign \maxRange($D$); $\;\;$ $x_{new}$ \assign \median($D$, $n_{new}$) \;
   $v_{new}$ \assign $(x_{new}, n_{new})$ ; $\;\;$
   $\mathcal V$ \assign $\mathcal V \cup v_{new}$\;
   \lIf(){$s\neq \emptyset$}{
   $\mathcal E^s$ \assign $\mathcal E^s \cup (v_i, v_{new})$
   }
   $L$ \assign $D \cap \goth h^-(x_{new}, n_{new})$;  $\;\;$
   $R$ \assign $D \backslash (L \cup x_{new})$\;
   \buildSubtree($L$, $v_{new}$, -); $\quad$\buildSubtree($R$, $v_{new}$, +)\;
   }
   \buildSubtree($X$, $\emptyset$, $\emptyset$)\;
   \Return $ \mathcal T = (\mathcal V, \mathcal E^-, \mathcal E^+)$\;
 }

 \end{algorithm}
\vspace{-9mm}
\subsection*{Incremental \kdtree build algorithm}\vspace{-2mm}
 A key operation to build a \kdtree incrementally is to pick splitting hyperplanes on the fly as new data points are inserted. This is achieved with a \emph{splitting sequence}, which we define as:

\vspace{2mm} \noindent{\bf Splitting sequence.} A map $\mathcal Z_\manifold: \natural_{\geq 0} \times \manifold \rightarrow \real^n$ that, given an integer $d \geq 0$ and a point $p \in \manifold$, returns a normal vector $n$ associated with them.

 When dealing with $k$-dimensional data, one usually chooses a basis of \emph{cardinal axes} $\{\hat e_i\}_{i\in[0\dots k-1]}$.
The normal of the hyperplane associated with a vertex $v$ is typically picked by cycling through the cardinal axes based on the depth of $v$ in the binary tree. Let ``$a \hspace{-1mm} \mod b$'' denote the remainder of the division of integer $a$ by integer $b$. Then, the \emph{classic splitting (CS) sequence} is defined as:
 \begin{equation}
  \mathcal Z_{\manifold}^{classic} (d, x) = \hat e_{[d\; mod\; k]}.
\end{equation}

\begin{algorithm}[!ht]
  \label{alg:incremental}
  \caption{Incremental \kdtree build. \newline Insert in \kdtree $\mathcal T = (\mathcal V, \mathcal E^+, \mathcal E^-)$, start recursion with  \hspace{-1mm} \texttt{insert}($x_{new}$, \texttt{root}($\mathcal T$)).}
\SetKwData{side}{side}
\SetKwFunction{sideOf}{sideOf}
\SetKwFunction{assign}{\ensuremath{\longleftarrow}}
\SetKwFunction{insert}{insert}
\SetKwFunction{dpth}{depth} 
\SetKwFunction{child}{child}

\SetKwBlock{fnblock}{define \insert($x_{new} \in \manifold$, $v_i = (x_i, n_i) \in \mathcal V$):}{end}

\fnblock{
\If(\tcp*[f]{if tree is empty...}){$\mathcal V  = \emptyset$}{
  $n_{new}$ \assign $\mathcal Z_\manifold(0, x_{new})$; $ \mathcal V \assign (x_{new}\,,\, n_{new}) $ \tcp*{...add root}
  \Return
}
\side \assign \sideOf($x_{new}, x_i, n_i$) ; c \assign \child($v_i$, \side) \;
\If(){$c = \emptyset$}{
  $n_{new}$ \assign $\mathcal Z_\manifold(\dpth(v_i)+1, x_{new})$ ;
  $v_{new}$ \assign $(x_{new}\,,\, n_{new}) $\;
  $\mathcal V$ \assign $\mathcal V \cup v_{new}$ ;
  $\mathcal E^{\side}$ \assign $\mathcal E^{\side} \cup (v_i, v_{new})$ \;
  \Return
}
\insert($x_{new}$, $c$) \;
}
\end{algorithm}
\vspace{-10mm}
\section{Complexity Analysis}\label{sec:analysis}

In this section, we discuss the expected asymptotic complexity of $m$-nearest-neighbor queries in a \kdtree built according to Algorithm \ref{alg:batch}.
The complexity of the nearest-neighbor query procedure (Algorithm \ref{alg:query}) is measured in terms of the number $n_v$ of vertices examined.
Let $\mathcal B(q, d_m)$ be the ball containing the $m$ nearest neighbors to $q$.
For the soundness of the algorithm, all the vertices contained in this ball must be examined, and therefore, all the buckets overlapping with it. As we recall from Section \ref{sec:kdtrees}, buckets are associated with leaves, and therefore the algorithm will visit as many leaves as the number of such buckets, formally: $n_l \in \Theta (|\beta|)$, where $ \beta \coloneqq \{b \in \mathscr B_{\mathcal T}\,|\, b\,\cap\,\mathcal B(q, d_m) \neq \emptyset\}$ and $n_l$ denotes the number of leaves visited by the algorithm.

If the asymptotic order of visited leaves is known, proofs rely on the following fact: since the batch algorithm guarantees a balanced tree, descending into $n_l$ leaves from the root requires visiting $n_v \in \Theta(n_l \log N)$ vertices, where $N$ is the cardinality of the database $X$. Thus the expected query complexity is given by:
\begin{equation}\label{eq:node_leaves}
  \mathbb E(n_v) \in \Theta [\mathbb E(n_l)\log N].
\end{equation}

Lemma \ref{lem:classic} reviews the seminal result, originally from \cite{friedman1977algorithm}, that the expected asymptotic complexity for a Minkowski (i.e., L-p) distance is logarithmic.
Our main theoretical contribution is the negative result described in Theorem \ref{thm:broken_classic}, where we reveal that the expected complexity for sub-Riemannian metrics is in fact super-logarithmic.
In the sequel, assume that data points are randomly sampled from $\manifold$ with probability distribution $p(x)$.

\begin{lemma}[\cite{friedman1977algorithm}]\label{lem:classic}
If the distance function used for the nearest-neighbor query is induced by a $p$-norm, then the expected complexity of the nearest-neighbor search on a \kdtree built with Algorithm \ref{alg:batch} is $O(\log(N))$ as $N\rightarrow\infty$.
\end{lemma}

For the detailed proof, we suggest reading the original work by Friedman et al. \cite{friedman1977algorithm}, while here we report key facts used in our next result.
In \cite[pg. 214]{friedman1977algorithm}, it is shown that:
\begin{equation}\label{eq:e_mball_volume}
  \mathbb E[\mathrm{vol}(\mathcal B(q, d_m))] \approx  \frac{m}{(N+1)}\frac{1}{p(q)},
\end{equation}
with no assumptions on the metric and the probability distribution $p(x)$.

In a batch-built \kdtree, the expected asymptotic shape of the bucket $\goth b_q$ containing the query point $q$ is assumed hyper-cubical.
From equation (\ref{eq:e_mball_volume}) with $m=1$, $\mathbb E(\mathrm{vol}(\goth b_q)) \approx \frac{1}{(N+1)}\frac{1}{p(q)}$.
Thus, hyper-cubes in a neighborhood of $q$ have sides of expected length $E(l_q)=[(N+1)p(q)]^{-1/k}$.
Using these facts, it is then shown that the expected number of visited leaves is asymptotically constant with the size of the database, $E(n_l) \in \Theta(1)$ and therefore the average query complexity is logarithmic, as per equation (\ref{eq:node_leaves}).

\begin{theorem} \label{thm:broken_classic}
Let the distance function used for the nearest-neighbor query be a sub-Riemannian metric on a smooth connected manifold $\manifold$ with a bracket-generating horizontal distribution of dimension $h$ and weights $\{ w_i\}$.
Then the expected complexity of the nearest-neighbor query on a \kdtree built with Algorithm \ref{alg:batch} is $\Theta(N^p\log(N))$ as $N\rightarrow\infty$, where the expression for $p$ is:
\begin{equation}\label{eq:complexity}
	 p = \hspace{-2mm} \sum_{w_i \leq W/k} \left(\frac{1}{k}-\frac{w_i}{W}\right), \;\;\text{with }\; W\coloneqq\sum_i w_i\,.
\end{equation}

\end{theorem}
\begin{proof}

For $N$ large enough, Theorem \ref{thm:ball_box} ensures the existence of inner and outer bounding boxes to the sub-Riemannian ball $\mathcal B(q, d_m)$.

From equations (\ref{eq:ball_box}) and (\ref{eq:e_mball_volume}), we get:%
\begin{equation}\label{eq:box_bound}
 \mathbb E(d_m)^{W} \cdot \prod_{i=1}^k c_i \leq  \frac{1}{N+1}\frac{1}{p(q)} \leq \mathbb E(d_m)^W \cdot \prod_{i=1}^k C_i,
\end{equation}
where $W \coloneqq \sum_i w_i$.
Therefore, the expected distance to the $m$-th nearest neighbor scales asymptotically as $\mathbb E(d_m)\in \Theta(N^{-1/W})$ for $N\rightarrow \infty$.
Recall that a bucket $\goth b_q$ has sides of expected length $l_{q}=\sqrt[k]{1/\left[(N+1)p(q)\right]}$ or equivalently, $l_{q} \in \Theta(N^{-1/k})$.
We are now interested in the asymptotic order of the number of buckets overlapping with a weighted box of size $d_m$.
Let $n_h(d, l_q)$ the number of hypercubes intersected by a segment of length $d$ embedded in a grid of $\mathbb{R}^k$ with side $l_q$.
It can be shown that $\left\lceil (\sqrt{k})^{-1} (d/l_{q}) \right\rceil \leq n_h(d, l_q) \leq k + \sqrt k \lceil d/l_q\rceil$. Then, asymptotically, $n_h(d, l_q) \in \Theta(\lceil d/l_q \rceil)$ as $(d, l_q) \rightarrow 0$.
Since a box has $k$ orthogonal sides, each with expected length $\mathbb E(d_m)^{w_i}$, the expected number of visited buckets, is:
\begin{equation}\label{eq:proof1}
  \mathbb E(n_l) \in
  \Theta\left(\prod_{i=1}^k \left\lceil \frac{\mathbb E(d_m)^{w_i}}{l_{q}} \right\rceil \right) = \Theta\left(\prod_{i=1}^k \left\lceil  N^{\frac{1}{k}-\frac{w_i}{W}}\right\rceil \right).
\end{equation}

In the latter product, only the factors with exponent $> 0$ contribute to the complexity of the algorithm, while the other terms tend to $1$, as $N\rightarrow \infty$.
In other words, the query complexity is determined only by low order Lie brackets up until $w_i \leq W/k$. In fact, along the direction of higher order Lie brackets, reachable sets shrink to zero faster than the side of a \kdtree bucket. Following up from equation (\ref{eq:proof1}), we get:
\begin{equation}
  \mathbb E(n_l) \in \Theta\left(\prod_{w_i \leq W/k}  N^{\frac{1}{k}-\frac{w_i}{W}}\right) = \Theta\left(N^p\right), \quad \; p = \hspace{-2mm} \sum_{w_i \leq W/k} \left(\frac{1}{k}-\frac{w_i}{W}\right).
\end{equation}
From equation (\ref{eq:node_leaves}) it follows that the expected complexity of the query algorithm is given by  $\mathbb E(n_v) \in \Theta(N^p \log N)$ as $N \rightarrow \infty$. \qed
\end{proof}

\paragraph{Remark.} For holonomic systems, the query complexity is logarithmic, in accordance with Lemma \ref{lem:classic}.
In fact, for holonomic systems, $w_1 = w_2 = \dots = w_k = 1$, then $W = k$ and $p = 0$ from equation (\ref{eq:complexity}).

\paragraph{Remark.} The query complexity is always between logarithmic and linear. Since $W \geq k$ by definition, for all $i$ such that $w_i \leq W/k$, one can state $0 < w_i/W \leq 1/k \leq 1 $. It follows that $0 \leq \frac{1}{k}-\frac{w_i}{W} < \frac{1}{k} $.
Then $p$ can be bounded with:
\begin{equation}
  0 \leq p = \hspace{-2mm}\sum_{w_i \leq W/k} \left(\frac{1}{k}-\frac{w_i}{W}\right) < \hspace{-2mm} \sum_{w_i \leq W/k} \left(\frac{1}{k}\right) \leq k \cdot \frac{1}{k} = 1 \;\; \Rightarrow \;\; p \in [0, 1).
\end{equation}

\begin{example-rs}

For the Reeds-Shepp car, $k=\mathrm{dim}[SE(2)] = 3$, $w_{\hat f} = w_{\hat \theta} = 1$, $w_{\hat l} = 2$ and $W = w_{\hat f} + w_{\hat \theta} + w_{\hat l} = 4.$ Only $w_{\hat f}$ and $w_{\hat \theta}$ satisfy $w_i \leq W/k$, therefore, according to Theorem \ref{thm:broken_classic}, the expected query complexity of a batch kd-tree is $\Theta\left(\sqrt[6]{N} \log N\right)$.
We confirm this prediction with experiments in section~\ref{sec:experiments}.
\end{example-rs}

\section{The Lie splitting strategy}\label{sec:lie_split}
The basic working principle of \kdtrees is the ability to discard subtrees without loss of information during a query.
For each visited node $v = (x,n)$, the algorithm checks whether the current biggest ball in the queue, $\mathcal B(q, d_m)$ intersects $\goth h(x,n)$.
If such an intersection exists, the algorithm visits both children of $v$, otherwise one child is discarded and so is the entire subtree rooted in it.
To reduce complexity, it is thus desirable that during query, a \kdtree presents as few ball-hyperplane intersections as possible.
 \begin{figure}[!ht]
   \centering\vspace{-0.8cm}
   \subfigure[Classic \kdtree construction \label{fig:kdclassic}]{\includegraphics[page=8, width=.4\textwidth]{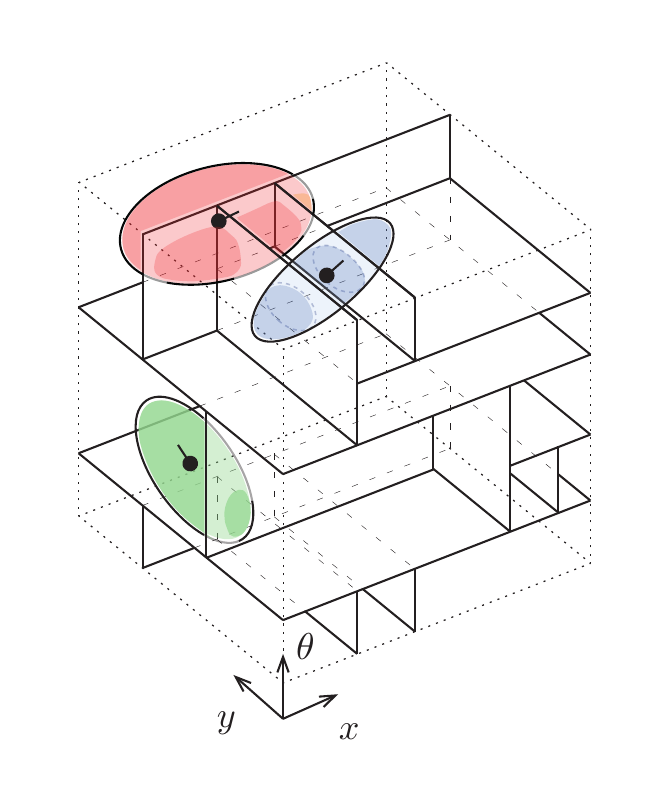}} \hfil
   \subfigure[Lie \kdtree splitting strategy \label{fig:kdlie}]{\includegraphics[page=12, width=.4\textwidth]{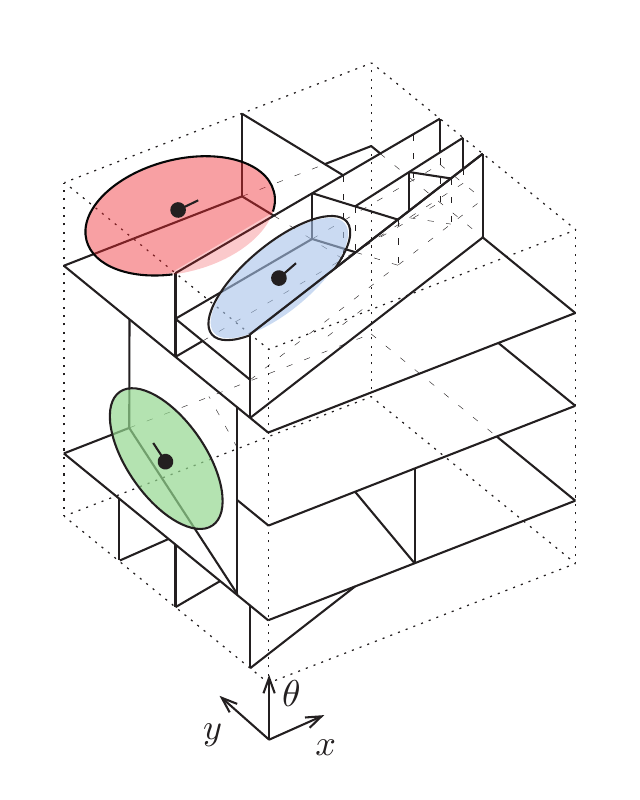}}
   \caption{\label{fig:liesplitting_seq} Qualitative comparison of a classic \kdtree with its counterpart built with the proposed \emph{Lie splitting strategy}. For nonholonomic systems, reachable sets (red, blue, green) are elongated along configuration-dependent \emph{privileged directions}. The smaller the sets, the more pronounced their aspect ratios. The Lie splitting strategy adapts the hyperplanes locally to the balls and decreases the expected number of ball-hyperplane intersections, thus reducing the expected asymptotic query complexity.}
 \end{figure}

In the classic splitting strategy, hyperplanes are chosen cycling through a globally defined set of cardinal axes.
However, we have shown that nonholonomic systems have a set of locally-defined \emph{privileged axes} and that the reachable sets have different infinitesimal orders along each of them.
When the metric comes from a nonholonomic system, the hyperplanes in a classic \kdtree are not aligned with reachable sets and the buckets have different asymptotic properties than reachable sets. This can make intersections frequent, as shown in figure \ref{fig:kdclassic}.

Ideally, to minimize the number of ball-hyperplane intersections, the buckets in a \kdtree should \emph{approximate the bounding boxes for the reachable sets} of the dynamical system, as depicted in figure \ref{fig:kdlie}.
To achieve this, we propose a novel splitting rule, named the \emph{Lie splitting strategy}, which exploits the differential geometric properties of a system and the asymptotic scaling of its reachable sets.\vspace{1mm}
The Lie splitting strategy is based on the following two principles:\vspace{-2mm}
\begin{enumerate}
  \item The splitting normal associated with each data point $x_i$ is along one of the privileged axes in that point, i.e., $\mathcal Z_\manifold(d, x_i) \in \{y_1(x_i), y_2(x_i) \dots y_k(x_i)\}$.
  \item The buckets of the \kdtree, in expectation, scale asymptotically according to the weighted box along all privileged axes as $t \rightarrow 0$.
  Formally, for all $q \in \manifold$ and for all pairs of privileged axes $y_i(q),y_j(q)$,
  \begin{equation}\label{eq:lie_asymptotic_bucket}
    \mathbb E\left[\range_{y_i(q)}(\goth b_q)\right]^{w_j} \in \Theta\left( \mathbb E\left[\range_{y_j(q)}(\goth b_q)\right]^{w_i} \right).
  \end{equation}
\end{enumerate}
Requirement 2 prescribes the asymptotic behavior of the sequence $\mathcal Z_\manifold$ as $d \rightarrow \infty$, which can be formalized as follows:
let $n_i(d) = \left|\{n < d: \mathcal Z_\manifold(n, x) = y_i(x)\}\right|$, i.e., the total number of splits in the sequence $\mathcal Z_\manifold$ along axis $y_i$ before index $d$.

As $d \rightarrow \infty$, the expected bucket size along $y_i$ after $n_i(d)$ splits has asymptotic order $\mathbb E\left[\range_{y_i(q)}(\goth b_q)\right] \in \Theta[e^{-n_i(d)}]$.
Then, in terms of the number of splits, equation (\ref{eq:lie_asymptotic_bucket}) yields:

\begin{equation}\label{eq:lie_splitting}
  n_i(d)\cdot w_j \sim n_j(d) \cdot w_i \text{ for all } i,j \in [1, \dots k] \text{ as } d\rightarrow \infty.
\end{equation}

Simply put, each privileged direction should be picked as a splitting normal with a frequency proportional to its weight. Note that this is only relevant \emph{asymptotically}, i.e., as $d\rightarrow\infty$.

\begin{example-rs}
  For the Reeds-Shepp vehicle, a valid Lie splitting sequence is:
\begin{equation}
  \mathcal Z_{RS}^{Lie}(d, x) \coloneqq \left\{
  \begin{array}{ll}
    \hat l(x), & \text{if} \; \; d\equiv0\,\text{ or }\, d\equiv 1 \quad(\mathrm{mod }\, 4)  \\
    \hat f(x), & \text{if} \; \; d\equiv2 \quad(\mathrm{mod }\, 4)  \\
    \hat \theta(x), &  \text{if} \; \; d\equiv3 \quad(\mathrm{mod }\, 4)   \\
  \end{array}\right.
\end{equation}
It is easy to verify that this satisfies equation (\ref{eq:lie_splitting}). In fact, lateral splits occur in the sequence twice as often as longitudinal and lateral splits.
\end{example-rs}
\section{Experimental results}\label{sec:experiments}

In this section, we validate our results and compare the performance of our algorithms with different methods from widely used open-source libraries.
Query performances are averaged over 1000 randomly drawn query points. The same insertion and query sequences are used across all the analyzed algorithms and generated from a uniform distribution.
\paragraph{Experiment 1.}

In this experiment, we confirm the theoretical contributions presented in Section \ref{sec:analysis}.
In figure \ref{fig:ex1} we show the average number of leaves visited when querying a batch \kdtree with Euclidean (blue) and Reeds-Shepp metrics (solid red). While the blue curve settles to a constant value, in accordance with Lemma~\ref{lem:classic}, the red curve exhibits exponential growth. When normalizing this curve by $\sqrt[6] N$ (dashed red), we observe a constant asymptotic behavior, consistent with the rate determined by Theorem~\ref{thm:broken_classic}.
For comparison, we report the average time for Euclidean queries on an incremental \kdtree (black), which tends to visit more vertices than its batch, balanced counterpart (blue).
\begin{figure}[!htb]\vspace{-6mm}\hspace{-4mm}
  \subfigure[\label{fig:ex1} Average leaves visited (query)]{\includegraphics[width=.492\textwidth]{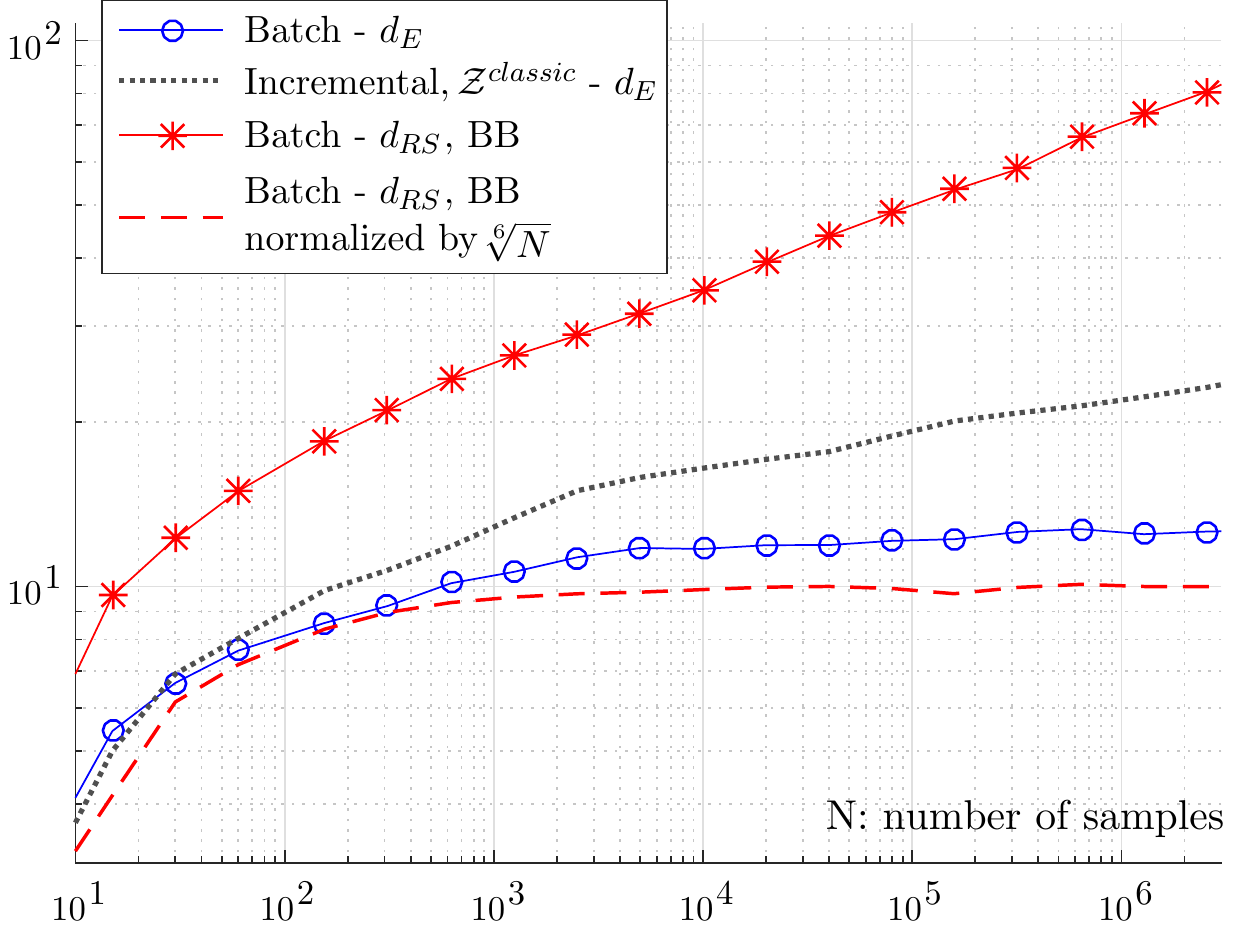}} \hfill
  \subfigure[\label{fig:ex2t} Total distance evaluations]{\includegraphics[page=1,width=.52\textwidth]{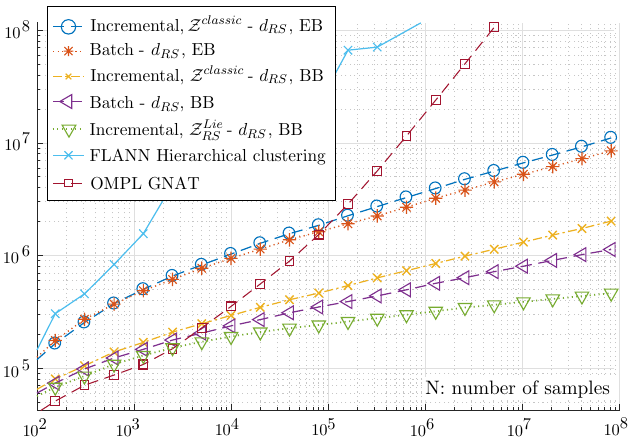}} \\
  \subfigure[\label{fig:ex2b} Build time $(\mu s)$]{\hspace{-3mm}\includegraphics[page=2,height=.276\textwidth]{exp2}} \hfill
  \subfigure[\label{fig:ex2q} Average query time $(\mu s)$]{\includegraphics[page=3,height=.276\textwidth]{exp2}} \\
  \vspace{-3mm}
  \caption{(a) \emph{Experiment 1}: Average number of leaves visited during Euclidean (blue) and Reeds-Shepp (red, solid) queries of a batch \kdtree. (b-d) \emph{Experiment 2}: Performance of different algorithms. In the legends, $d_E$ and $d_{RS}$ indicate Euclidean and Reeds-Shepp queries, $EB$ and $BB$ indicate Euclidean Bound and outer Box Bound for ball-hyperplane intersections, $\mathcal Z^{classic}$ and  $\mathcal Z^{Lie}_{RS}$ indicate the splitting sequence. \vspace{-4mm}}
\end{figure}
\vspace{-2mm}
\paragraph{Experiment 2.}
In this experiment (figures \cref{fig:ex2t,fig:ex2b,fig:ex2q}), we plot the total number of distance evaluations and the running times observed with different combinations of build and query algorithms.
Figure \ref{fig:ex2q} reveals that the proposed outer Box Bound (BB) ball-hyperplane intersection method~(yellow, purple, green) reduces the query time significantly as compared to Euclidean bounds~(EB)~(blue, red), in accordance with our predictions in Section \ref{sec:query}.

Additionally, Lie splitting (green) further improves query time, as compared with the classic splitting (yellow). The corresponding incremental Lie \kdtree
also outperforms a classic batch-built one (purple), guaranteed to be balanced.

More important speedups emerge when comparing the proposed \kdtrees with different techniques, such as Hierarchical Clustering from FLANN and Geometric Near-neighbor Access Tree (GNAT) from OMPL. Interestingly, off-the-shelf implementations of \kdtrees offered by FLANN and other tools are unusable with nonholonomic metrics altogether, since they are limited to distances of the form of equation (\ref{eq:accum_dist}). Therefore a comparison is not possible.

All the tested \kdtrees visibly outperform Hierarchical Clustering~(cyan) both in build time and query time. In contrast, GNAT offers competitive query times. However, its insertion is $\sim 100 \times$ slower than incremental k-d trees, since a significant number of distances are evaluated in the build phase, while \kdtrees only evaluate distances during query. This is reflected in a noticeably higher asymptotic rate of distance evaluations, revealed in figure \ref{fig:ex2t}.

\vspace{-1mm}
\section{Conclusion}
Motivated by applications in sampling-based motion planning, we investigated \kdtrees for efficient nearest-neighbor search with distances defined by the length of paths of controllable nonholonomic systems. We have shown that for sub-Riemannian metrics, the query complexity of a classic batch-built \kdtree is $\Theta(N^p \log N)$, where $p \in [0,1)$ depends on the properties of the system.
In addition, we have proposed improved build and query algorithms for \kdtrees that account for differential constraints. The proposed methods proved superior over classic ones in numerical experiments carried out with a Reeds-Shepp vehicle.

Future work will analyze whether logarithmic complexity is achieved for nonholonomic systems. In addition, the proposed algorithms are exact and rely on explicit distance evaluations. Since distances cannot be generally computed in closed form, we are interested in investigating approximate nearest-neighbor search algorithms with provable correctness bounds that do not require explicit distance computations.

\bibliographystyle{abbrv}
\bibliography{references}\vspace{-5mm}

\end{document}